\documentclass[11pt]{article}
\usepackage[compact]{titlesec}

\usepackage{times}
\usepackage{latexsym}
\usepackage{amsfonts,amsthm,amssymb}
\usepackage{amsmath}
\usepackage{euscript}
\usepackage{amstext}
\usepackage{graphicx}
\usepackage{color}

\usepackage{url,hyperref}
\usepackage{verbatim}
\usepackage{xspace}
\usepackage{algorithm2e}
\usepackage{framed}

\newtheorem{lemma}{Lemma}[section]
\newtheorem{theorem}[lemma]{Theorem}

\newtheorem{corollary}[lemma]{Corollary}
\newtheorem{proposition}[lemma]{Proposition}

        {\hspace*{\fill}$\Box$\par}

\newcommand{\eps}{\epsilon}
\newcommand{\setf}{\textrm{\sc SETF}\xspace}
\newcommand{\srpt}{\textrm{\sc SRPT}\xspace}

\newcommand{\rr}{\textrm{\sc  RR}\xspace}
\newcommand{\wrr}{\textrm{\sc  WRR}\xspace}

\newcommand{\laps}{\textrm{\sc  LAPS}\xspace} 
\newcommand{\wlaps}{\textrm{\sc  WLAPS}\xspace}

\newcommand{\J}{{\cal J}}

\newcommand{\initOneLiners}{%
    \setlength{\itemsep}{0pt}
    \setlength{\parsep }{0pt}
    \setlength{\topsep }{0pt}
}

\renewcommand{\S}{\mathcal{J}}

\newcommand{\W}{\mathcal{W}}
\renewcommand{\P}{\mathcal{P}}

\newcommand{\lpp}{\mbox{$\mathsf{LP}_\mathsf{primal}$}\xspace}
\newcommand{\lpd}{\mbox{$\mathsf{LP}_\mathsf{dual}$}\xspace}

\newcommand{\cpp}{\mbox{$\mathsf{CP}_\mathsf{primal}$}\xspace}
\newcommand{\cpd}{\mbox{$\mathsf{CP}_\mathsf{dual}$}\xspace}

\newcommand{\slaps}{\textrm{\sc WRPS}\xspace} 

\begin{document}
\title{{\sc SelfishMigrate}: A Scalable Algorithm for Non-clairvoyantly Scheduling Heterogeneous Processors}
\author{Sungjin Im\thanks{ Electrical Engineering and Computer Science Department, University of California, 5200 N. Lake Road, Merced CA 95344. {\tt sim3@ucmerced.edu}.}  \and Janardhan Kulkarni \thanks{Department of Computer Science, Duke University , 308 Research Drive, Durham, NC 27708. {\tt kulkarni@cs.duke.edu}. Supported by NSF Awards CCF-1008065 and IIS-0964560.}  \and
Kamesh Munagala\thanks{Department of Computer Science, Duke University, Durham NC 27708-0129. {\tt kamesh@cs.duke.edu}. Supported by an award from Cisco, and by NSF  grants CCF-0745761, CCF-1008065, CCF-1348696, and IIS-0964560.} \and
Kirk Pruhs\thanks{Department of Computer Science, University of Pittsburgh, Pittsburgh, PA 15260. {\tt  kirk@cs.pitt.edu}.Supported in part by NSF grants CCF-1115575, CNS-1253218, and an IBM  Faculty Award. }
}

\date{}
\maketitle

\thispagestyle{empty}

\begin{abstract}
We consider the classical problem of minimizing the total weighted flow-time for unrelated machines in the online \emph{non-clairvoyant} setting. In this problem, a set of jobs $J$ arrive over time to be scheduled on a set of $M$ machines. Each job $j$ has processing length $p_j$, weight $w_j$, and is processed at a rate of $\ell_{ij}$ when scheduled on machine $i$. The online scheduler knows the values of $w_j$ and $\ell_{ij}$ upon arrival of the job, but is not aware of  the quantity $p_j$.  We present the {\em first} online algorithm that  is {\em scalable} ($(1+\eps)$-speed $O(\frac{1}{\epsilon^2})$-competitive for any constant $\eps > 0$) for the total weighted flow-time objective. No non-trivial results were known for this setting, except for the most basic case of identical machines. Our result resolves  a major open problem in online scheduling theory.  Moreover, we also show that no job needs more than a logarithmic number of migrations.

We further extend our result and give a scalable algorithm for the objective of minimizing total weighted flow-time plus energy cost for the case of unrelated machines. In this problem, each machine can be sped up by a factor of $f^{-1}_i(P)$ when consuming power $P$, where $f_i$ is an \emph{arbitrary} strictly convex power function.  In particular, we get an $O(\gamma^2)$-competitive algorithm when all power functions are of form $s^{\gamma}$. These are the {\em first} non-trivial non-clairvoyant results in any setting with heterogeneous machines. 

The key algorithmic idea is to let jobs migrate selfishly  until they converge to an equilibrium. Towards this end, we define a  game where each job's utility which is closely tied to the instantaneous increase in the objective the job is responsible for, and each machine declares a policy that assigns priorities to jobs based on when they migrate to it, and the execution speeds. This has a spirit similar to coordination mechanisms that attempt to achieve near optimum welfare in the presence of selfish agents (jobs).  To the best our knowledge, this is the first work that demonstrates the usefulness of ideas from coordination mechanisms and Nash equilibria for designing and analyzing online algorithms. 
\end{abstract}

\newpage
\setcounter{page}{1}

\section{Introduction}
	\label{sec:intro}
Many computer architects believe that architectures consisting of heterogeneous processors will be the dominant architectural design in the future: Simulation studies indicate that, for a given area and power budget,  heterogeneous multiprocessors can offer an order of magnitude better performance for typical workloads~\cite{Bower08, KumarTJ06, Merritt10, Koufaty10}. Looking at the consequences of Moore's Law even further in the future, some computer architectures are projecting that we will transition from the current era of multiprocessor scaling to an era of ``dark silicon'', in which switches become so dense that it is not economically feasible to cool the chip if all switches are simultaneously powered.~\cite{EsmaeilzadehBASB11}.
One possible architecture in the dark silicon era would be many specialized processors, each designed for a particular type of job. The processors that are on any point of time should be those that are best suited for the current tasks.

It is recognized by the computer systems community~\cite{Bower08} and
the algorithms community that scheduling these future
heterogeneous multiprocessor architectures is a major challenge.
It is known that some of the standard scheduling algorithms for single processors and homogeneous processors
can perform quite badly on heterogeneous processors~\cite{GuptaIKMP12}. 
A scalable algorithm (defined below) is known if somehow the scheduler was clairvoyant (able to know
the size of a job when it arrives)~\cite{ChaddaGKM}; however, this knowledge is generally not
available in general purpose computing settings. 
A scalable algorithm is also known if all jobs were of equal importance~\cite{ImKM14};
however, the whole raison d'\^etre for heterogeneous architectures is that there is
generally heterogeneity among the jobs, most notably in their \emph{importance/priorities}.

\medskip
Therefore, a major open question in the area of online scheduling, both in theory and practice, is the design of scalable online algorithms for scheduling heterogeneous processors with arbitrary power functions, when job sizes are not known in advance (non-clairvoyant setting), and jobs have different priorities (weighted setting). The typical objective that has been studied in this context is minimizing weighted delay, or weighted flow-time. This problem generalizes most previous work~\cite{AnandGK12,ChaddaGKM,Devanur014,ImKM14,GuptaIKMP12,ImKM14} in online single and multiple machine scheduling considered recently by the algorithms community. As we discuss below, the algorithmic techniques developed for these problems (clairvoyant setting, unweighted setting, etc) do not extend in any natural way to the most general setting with weights, non-clairvoyance, and energy, leading to a fundamental gap in our understanding of classical scheduling algorithms. In particular, we ask: 
\begin{quote}
Do different variants of multiple machine scheduling considered in literature require different algorithmic techniques and analyses, or is there one unifying technique for them all?  %in the non-clairvoyant setting?
\end{quote}

In this paper, we close this gap, and obtain {\em one unifying algorithmic technique for them all},  achieving the {\em first} scalable non-clairvoyant algorithms for scheduling jobs of varying importance on heterogeneous machines, even with arbitrary power functions. % first obtain an algorithm for the case that processors can not be shutdown, in which case the objective is to optimize performance, as energy usage is fixed.  We then obtain an algorithm for the case that processors can be shutdown when  not in use to save energy, in which case the objective is to optimize a specified energy/performance tradeoff.  
The interesting aspect of our work, as we discuss in Section~\ref{sec:tech} below, is that it provides an algorithmically simple and conceptually novel framework for multiple machine scheduling as a coordination game (see \cite{CoordinationMechanisms,SB}), where jobs have (virtual) utility functions for machines based on delays they contribute to, and machines announce scheduling policies, and treat migration of jobs into them as new arrivals. In hindsight, we believe this provides the correct, unifying way of viewing all recent algorithmic results~\cite{AnandGK12,Devanur014,ImKM14,GuptaIKMP12} in related and simpler models.

\subsection{Our Results}
We adopt the competitive analysis framework in online algorithms. %When any online algorithm has a large competitive ratio, so that the competitive ratio gives little insight into the quality of algorithms, a popular relaxation is speed augmentation~\cite{kirk} where the online scheduler is given extra speed and compared against a slightly weaker optimal scheduler without the extra speed.  
We say that an online schedule is $s$-speed $c$-competitive  if it is given $s$ times faster machines and is $c$-competitive when compared to the optimal scheduler with no speed augmentation. The goal is to design a $O(1)$-competitive algorithm with the smallest extra speed. In particular, a {\em scalable algorithm}, which is $(1+\eps)$-speed $O(1)$-competitive for any fixed $\eps >0$, is considered to be essentially the best result one can hope for in the competitive analysis framework for machine scheduling~\cite{kirk}. (It is known that without resource augmentation no online algorithm can have bounded competitive ratio for our problem \cite{GargK07, ChaddaGKM}.)

\paragraph{Unrelated Machine Scheduling.}
We first consider the general \emph{unrelated} machine model where each job $j$ can be processed at rate $\ell_{ij} \geq 0$ on each machine $i$; if $\ell_{ij}=  0$, then job $j$ cannot be processed on machine $i$. In a feasible schedule, each job can be scheduled by at most one machine at any point in time. Preemption is allowed without incurring any cost, and so is migration to a different machine.  A job $j$'s flow-time $F_j := C_j  - r_j$ measures the length of time between when the job arrives at time $r_j$ and when the job is completed at $C_j$. When each job $j$ has weight $w_j$, the total weighted flow-time is $\sum_j w_j F_j$; the unweighted case is where all jobs weights are one. The online scheduler does not know the job size $p_j$ until the job completes. For a formal problem statement, please refer to Section~\ref{sec:unrelated}. 

Our first result is the following. For comparison, we note that no constant speed, constant competitive result was known for the weighted case, except when all machines are identical.

\begin{theorem} 
	\label{thm:flow-no-scalable} [Section~\ref{sec:unrelated}.]
	For any $\eps >0$, there is a $(1+\eps)$-speed $O(1/ \eps^2)$-competitive non-clairvoyant algorithm for the problem of minimizing the total weighted flow-time on unrelated machines. 
	Furthermore, each job migrates at most  $O((\log W + \log n)/ \eps)$ times, where $W$ denotes the ratio of the maximum job weight to the minimum.
\end{theorem}

 It is well known~\cite{GuptaIKMP12} that in the non-clairvoyant setting, jobs need to be migrated to obtain $O(1)$ competitive ratio. Our algorithm migrates each job  relatively small number of times. Reducing the number of migrations is not only theoretically interesting but also highly desirable in practice~\cite{ChanLLZ13, GuptaKS14}. 
 %We note that the best known clairvoyant counterpart result is an $(1+\eps)$-speed $O(1 / \eps)$-competitive algorithm~\cite{AnandGK12, Devanur014}. Our algorithm uses a variant of weighted round robin scheduling on each machine -- the variant we use is a smooth extension of the algorithm Latest Arrival Processor Sharing (\laps or \wlaps) \cite{EdmondsP12, EdmondsIM11}. The best known competitive ratio for $\laps$ is $(1 / \eps^2)$ with $(1+\eps)$-speed. This shows our analysis is tight unless one gives a better analysis of $\laps$. 

\paragraph{Power Functions.}
We further extend our result to unrelated machines with arbitrary power functions.  Each machine $i$ is associated with an arbitrary power function $f_i: [0, \infty) \rightarrow [0, \infty)$ which is strictly convex and has $f_i(0) = 0$. When machine $i$ uses power $P >0$, it runs $f_i^{-1}(P)$ times faster than its original speed  --  it processes job $j$ at a rate of $\ell_{ij} f_i^{-1}(P)$. This model is widely studied~\cite{Albers10,  NH2009, GuptaKP10green, ChaddaGKM, GuptaIKMP12, Devanur014} and we consider the standard objective of minimizing the total weighted flow-time plus the total energy consumption~\cite{AlbersF07}. 

%We give the first non-clairvoyant algorithm for the average weighted flow-time objective that is scalable for unrelated machines.  The best corresponding clairvoyant result is a $(1+\eps)$-speed $O(1 / \eps)$-competitive algorithm~\cite{Devanur014}. 

\begin{theorem} 
	\label{thm:energy-scalable} [Section~\ref{sec:energy}.]
	For any $\eps >0$, there is a $(1+\eps)$-speed $O(1/ \eps^2)$-competitive non-clairvoyant algorithm for the problem of minimizing the total weighted flow-time plus total energy consumption on unrelated machines. This result holds even when each machine $i$ has an arbitrary strictly-convex power function $f_i: [0, \infty) \rightarrow [0, \infty)$ with $f_i(0) = 0$.   
\end{theorem}

The same guarantee on the number of migrations by a job stated in Theorem~\ref{thm:flow-no-scalable} holds for this setting.  The theorem implies a $O(\gamma^2)$-competitive algorithm (without resource augmentation) when each machine $i$ has a power function $f(s) = s^{\gamma}$ for some $\gamma >1$, perhaps most important power functions in practice. Our result also implies a scalable algorithm in the model where at each time instant a processor $i$ can either run at speed $s_i$ consuming a power $P_i$, or be shutdown and consume no energy. 

We note that no $O(1)$-speed $O(1)$-competitive non-clairvoyant algorithm was known prior to our work even in the related machine setting for {\em any} nontrivial classes of power functions. 

\iffalse
The next theorem immediately follows by simulating $1+ \eps = 1+ 1 / \gamma$ speed augmentation by using extra power. The same guarantee on the frequency of job migration carries to this special case as well. We note that there is a $O(\gamma/ \log \gamma)$-competitive algorithm when the scheduler is clairvoyant \cite{Thang13, Devanur014}. 

\begin{theorem} 
	\label{thm:energy-poly-scalable} [Section~\ref{sec:energy}.]
	Consider the case where each machine $i$ has a power function $f_i(s) = s^{\gamma_i}$ for some $\gamma_i >1$. Let $\gamma := \max_i \gamma_i$. Then, there is a a $O(\gamma^2)$-competitive non-clairvoyant algorithm for minimizing the total weighted flow-time plus total energy consumption on unrelated machines. 
\end{theorem}
\fi

%We note that all our algorithms are distributed and only involves each job knowing the total weight of jobs assigned to a particular machine.

\subsection{Technical Contributions: Selfish Migration and Nash Equilibrium} 
\label{sec:tech}
Our main technical contribution is a new conceptually simple game-theoretic framework for multiple machine scheduling that unifies, simplifies and generalizes previous work, both in terms of algorithm design as well as analysis using dual fitting. Before presenting this framework, we present some difficulties an online scheduler has to overcome in the non-clairvoyant settings we consider.

An online scheduler for multiple machines consists of two scheduling components: A single-machine scheduling policy on each machine, and the global machine assignment rule which assigns jobs to machines. In the context of clairvoyant scheduling~\cite{ChaddaGKM,AnandGK12}, the authors in~\cite{AnandGK12} show via a dual fitting analysis that the following algorithm is scalable: Each machine runs a scalable single-machine scheduling policy such as Highest Density First (HDF);  this is coupled with a simple greedy dispatch rule that assigns arriving jobs to the machine on which they cause the least increase in flow-time to previously assigned jobs. This simple yet elegant algorithm has been very influential. In particular, the greedy dispatch rule has become standard, and been used in various scheduling settings~\cite{ImM11,GuptaKP10,AnandGK12,Thang13,Devanur014}.  
The analysis proceeds by setting dual variables corresponding to a job to the marginal increase in total delay due to the arrival of this job, and showing that this setting is not only feasible but also extracts a constant fraction of the weighted flow-time in the objective. {\em The immediate-dispatch greedy rule is necessary for analysis in all the aforementioned work}, since they require the algorithm to measure each job's effect on the system at the moment it arrives. 
% xxx
%Such an approach, however, crucially needs an estimate of the impact an arriving job has on the delay of previous jobs, since these define the setting of dual variables.  This also leads to fixing the assignment of jobs to machines on arrival (no migration). 

In a non-clairvoyant context, there are main two hurdles that arise. First, to use the greedy dispatch rule, it is crucial to measure how much a new job affects the overall delay, that is, how much the job increases the objective. To measure the increase, the scheduler must know the job size, which is not allowed in the non-clairvoyant setting. Secondly, as mentioned before, jobs must migrate to get a $O(1)$-competitive algorithm even with any $O(1)$-speed, and this makes it more difficult to measure how much each job is responsible for the overall delay. This difficulty appears in the two analysis tools for online scheduling, potential function \cite{ImMP11} and dual fitting method \cite{AnandGK12,GuptaKP12,Devanur014}. Due to these difficulties, there have been very few results for non-clairvoyant scheduling on heterogeneous machines~\cite{GuptaKP10green,GuptaIKMP12,ImKM14}. Further, there has been \emph{no} work in \emph{any} heterogeneous machines setting for the weighted flow-time objective.

\subsubsection{{\sc SelfishMigrate} Framework}
We demonstrate a simple framework {\sc SelfishMigrate} that addresses the above two issues in one shot. Our algorithm can be best viewed in a  game theoretic setting where jobs are selfish agents, and machines declare their scheduling policies in advance. 

\paragraph{Machine Behavior.} Each machine maintains a {\em virtual queue} on the current set of jobs assigned to it; newly arriving jobs are appended to the tail of this queue.  In a significant departure from previous work~\cite{AnandGK12,Devanur014,GuptaKP10green,GuptaIKMP12,ImKM14}, each machine treats a migration of a job to it as an arrival, and a migration out of it as a departure. This means a job migrating to a machine is placed at the tail of the virtual queue. 

Each machine runs a scheduling policy that is a modification of weighted round robin (WRR) that smoothly assigns larger speed to jobs in the tail of the queue, taking weights into account. This is a smooth extension of the algorithm Latest Arrival Processor Sharing (\laps or \wlaps) \cite{EdmondsP12, EdmondsIM11}. We note that the entire analysis also goes through with WRR, albeit with $(2+\eps)$-speed augmentation. The nice aspect of our smooth policies (unlike \wlaps) is that we can approximate the instantaneous delay introduced by this job to jobs ahead of it in its virtual queue, {\em even without knowing job sizes}. This will be critical for our analysis.\footnote{The best known competitive ratio for $\laps$ is $O(1 / \eps^2)$ with $(1+\eps)$-speed, which shows our overall analysis is tight unless one gives a better analysis of $\laps$. }

\paragraph{Job Behavior.} Each job $j$ has a virtual utility function, which roughly corresponds to the inverse of the instantaneous weighted delay introduced by $j$ to jobs ahead of it in its virtual queue, and their contribution to $j$'s weighted delay.  Using these virtual utilities, jobs perform sequential best response (SBR) dynamics, migrating to machines (and get placed in the tail of their virtual queue) if doing so leads to larger virtual utility. Therefore, at each time instant, the job migration achieves a Nash equilibrium of the SBR dynamics on the virtual utilities. We show that our definition of the virtual utilities implies they never decrease due to migrations, arrivals, or departures, so that at any time instant the Nash equilibrium exists and is computable. (We note that at each time step, we simulate SBR dynamics and migrate each job directly to the machine that is predicted by the Nash equilibrium.)

When a job migrates to a machine, the virtual utility starts off being the same as the real speed the job receives. As time goes by, the virtual queue ahead of this job shrinks, and that behind it increases. This lowers the real speed the job receives, but its virtual utility, which measures the inverse of the impact to jobs ahead in the queue and vice versa, does not decrease. Our key contribution is to define the coordination game on the virtual utilities, rather than on the actual speed improvement jobs receive on migration. A game on the latter quantities (utility equals actual speed) need not even admit to a Nash equilibrium.

\medskip
Given the above framework, our analysis proceeds by setting the dual variable for a job to the increase in overall weighted delay it causes on jobs ahead of it in its virtual queue.  Our key insight is to show that Nash dynamics on virtual utilities directly corresponds to our setting of dual variables being feasible for the dual constraints, implying the desired competitive ratio.  This overall approach is quite general, even extending to energy constraints, and requires two key properties from the virtual utility:

\begin{itemize}
\item The virtual utility should correspond roughly to the inverse of the instantaneous delay induced by a job on jobs ahead of it in its virtual queue.  
\item SBR dynamics should monotonically improve virtual utility, leading to a Nash equilibrium that corresponds exactly to satisfying the dual constraints.
\end{itemize}

Our main (and perhaps quite surprising) contribution is to show the existence of such a virtual utility function for WRR and its scalable modifications, when coupled with the right notion of virtual queues. In hindsight, we believe this framework is the right way to generalize the greedy dispatch rules and dual fitting analysis from previous works~\cite{AnandGK12,ImKM14}, and we hope it finds more applications in complex scheduling settings.

\subsubsection{Comparison with Previous Techniques}
As mentioned above, the algorithmic idea based on coordination games is very different from the previous greedy dispatch rules~\cite{ChaddaGKM,AnandGK12,Devanur014} for clairvoyant scheduling, and also from previous work on non-clairvoyant scheduling~\cite{GuptaKP10green,GuptaIKMP12,ImKM14}. We contrast our algorithm with these, highlighting the necessity of new techniques. 

It is instructive to compare this framework with the scalable non-clairvoyant algorithm for {\em unweighted} flow-time~\cite{ImKM14}. This algorithm has the seeds of several ideas we develop here -- it introduces virtual queues, a smooth variant of \laps~\cite{EdmondsP12, EdmondsIM11} for single-machine scheduling, as well as migration based on the delay a job contributes to jobs ahead in its queue. However, this algorithm always assigns priorities to jobs in order of original arrival times, and migrates jobs to machines preserving this ordering. In essence, this algorithm mimics the clairvoyant algorithms~\cite{ChaddaGKM,AnandGK12} that account for delay a job contributes to those that arrived ahead in time.  This idea of arrival time ordering is specific to unweighted jobs, and does not extend to the weighted case or to energy constraints. In contrast, we let each machine treat migrations as new arrivals, leading us to view job migration through a game-theoretic lens. This leads to a more natural framework via instantaneous Nash equilibria, with a simple dual fitting analysis. The simplicity makes the framework extend easily to energy constraints.   The resulting accounting using the delay a job $j$ induces to those ahead of it in its virtual queue is novel -- in contrast to previous work~\cite{AnandGK12,Devanur014,GuptaKP10green,GuptaIKMP12,ImKM14}, the virtual queue changes with time and could possibly include jobs whose original arrival time $r_{j'}$ is later than that of job $j$. %Despite this apparent complexity, our overall framework, in hindsight provides the right way to think about and generalize previous work.

To further illustrate the technical difficulty of the weighted case, let us consider Round-Robin (\rr) and its variants. The work of~\cite{GuptaKP10green} consider unweighted jobs, and gives $(2+\eps)$-speed $O(1)$-competitive algorithm for the  total (unweighted) flow-time on related machines. The work of~\cite{GuptaIKMP12} improves  this to a scalable algorithm, but for simpler illustration, we focus on \rr. In the \rr used in~\cite{GuptaKP10green}, each of $n$ fastest machines is distributed to all $n$ jobs uniformly.  It is not difficult to see that this fractional schedule can be decomposed into a convex combination of feasible actual schedules. Hence, \rr allows us to view multiple machines as a more powerful single machine, and this is the crux of the analysis in \cite{GuptaKP10green}. In contrast,  the work of~\cite{GuptaIKMP12} argues that the weighted case is not as simple: in fact, they show that natural extensions of weighted round robin for related machines fail to be competitive.

%The work of~\cite{ImKM14} presents a scalable algorithm for non-clairvoyant scheduling of heterogeneous machines, but with unweighted jobs. This algorithm always assigns priorities to jobs in order of arrival times on the machines, and sequentially migrates jobs to machines using these queues. This idea of arrival time ordering is specific to unweighted jobs, and does not extend to the weighted case or to energy constraints. In contrast, we show that viewing migration from a game-theoretic lens with a virtual utility function, coupled with treating migration as the same as a new arrival from the perspective of a machine, leads to a more natural and general framework via instantaneous Nash equilibria, which also extends to energy constraints. 

\subsection{Other Related Work} 
For a survey on online scheduling, please see~\cite{PST}. As alluded to above, for weighted flow-time on unrelated machines, the best clairvoyant result is a $(1+\eps)$-speed $O(1 / \eps)$-competitive algorithm~\cite{AnandGK12, Devanur014}. For the version with power functions, the best corresponding clairvoyant result is a $(1+\eps)$-speed $O(1 / \eps)$-competitive algorithm~\cite{Devanur014}.   In the most basic setting of multiple machines where machines are all identical, the work of~\cite{ChekuriGKK04} gives the first analysis of scalable algorithms for weighted flow-time. Intuitively, machine assignment rule should have a spirit of load balancing. Indeed,  the work of~\cite{ChekuriGKK04}  shows two machine assignment rules can be combined with various single machine scheduling policies to yield a scalable algorithm. One is random assignment rule, and the other is based on volume of each job class. 

For the problem of non-clairvoyantly scheduling a single machine, the \wlaps\ (Weighted Latest Arrival Processor Sharing) algorithm~\cite{EdmondsP12, BansalKN10, EdmondsIM11} is scalable for the total weighted flow even when jobs have arbitrary speedup curves. Other non-clairvoyant scalalble algorithms for the unweighted case include Shorted Elapsed Time First (\setf) and Multi-level Feedback~\cite{kirk}. The work of~\cite{BansalP04} extends \setf to its weighted version. While Shortest Remaining Procesing Time (\srpt) is optimal for the total flow time objective on a single machine with clairvoyance, even slightly more general settings (non-clairvoyant or weighted or multiple machine settings) do not admit a $O(1)$-competitive algorithm without resource augmentation~\cite{BecchettiL04,BansalC09,LeonardiR07}.

\section{Unrelated Machine Scheduling}
\label{sec:unrelated}
In this problem, there are $m$ unrelated machines. Job $j$ is processed at rate $\ell_{ij}  \in [0, \infty)$ on each machine $i$.  Each job has processing length $p_j$ and weight $w_{j}$. The online algorithm is allowed to preempt and migrate jobs at any time with no penalty. The important constraint is that at any instantaneous time, a job can be processed only on a single machine. Job $j$ is released at time $r_j$.  In the non-clairvoyant online scheduling model we consider, the scheduler knows the values of $\ell_{ij}$ and $w_{j}$ when the job arrives, but is not aware of the processing length $p_j$.\footnote{It is easy to show that if $\ell_{ij}$ values are not known, then no online algorithm can have a bounded competitive ratio even with any constant speed augmentation.} Without loss of generality we assume that weights $w_j$ are integers. 

\medskip
Fix some scheduling policy $\P$. At each time instant $t$, each active job $j$ with $r_j \le t$ is assigned to some machine $i$. Let $J_i(t)$ denote the set of jobs assigned to machine $i$. Machine $i$ splits its processing power among the jobs in $J_i(t)$. Let $\nu_j(t)$ denote the processing power assigned to job $j \in J_i(t)$. We enforce that $\sum_{j \in J_i(t)} \nu_j(t) \le 1$ for all $i,t$. Then, $j \in J_i(t)$ executes at rate $q_j(t) = \ell_{ij} \nu_j(t)$. The completion time $C_j$ is defined as the earliest time $t_j$ such that
$$C_j = \mbox{argmin}_{t_j} \left(\int_{t = r_j}^{t_j} q_j(t) dt \ge p_j \right)$$ 
At this time, the job finishes processing and departs from the system. The objective is to find a scheduling policy that minimizes the sum of weighted flow-times $\sum_j w_j F_j$, where $F_j = C_j-r_j$ is the flow-time of job $j$.

\medskip
In the {\em speed augmentation} setting, we assume the online algorithm can process job $j$ at rate $(1+\epsilon) \ell_{ij}$ on machine $i$, where $\epsilon > 0$. We will compare the resulting flow-time against an offline optimum that knows $p_j$ and $r_j$ at time $0$, but is not allowed the extra speed. Our main result is a {\em scalable} algorithm that,  for any $\epsilon > 0$, is $O(1/\epsilon^2)$ competitive with speed augmentation of $(1+\epsilon)$. 

\iffalse We prove following theorem in this section.

\begin{theorem} 
	\label{thm:flow-no-scalable} [Section~\ref{sec:unrelated}.]
	For any $\eps >0$, there is a $O(1/ \eps^2)$-competitive non-clairvoyant algorithm with $(1+\eps)$-speed augmentation for the problem of minimizing the total weighted flow-time on unrelated machines.
\end{theorem}
\fi

\subsection{The {\sc SelfishMigrate} Algorithm}
Our algorithm can be best viewed as a coordination mechanism between the machines and the jobs. Each machine declares a {\em single machine policy} that it uses to prioritize and assign rates to arriving jobs. Given these policies, jobs migrate to machines that give them the most instantaneous utility (in a certain virtual sense). We will now define the single machine scheduling policy, and the utility function for jobs.

\subsubsection{Single Machine Policy: Weighted Ranked Processor Sharing, \slaps$(k)$}
This policy is parametrized by an integer $k$ (that we will later set to $1/\epsilon$) and $\eta > 1$ that captures the speed augmentation (and that we set later to $1+3\epsilon$). Fix some machine $i$ and time instant $t$. Recall that $J_i(t)$ denotes the set of jobs assigned to this machine at time $t$. Let $W(i,t)$ denote their total weight, {\em i.e.}, $W(i,t) = \sum_{j \in J_i(t)} w_j$.  The machine maintains a {\em virtual queue} on these jobs.

We now introduce some notation based on these virtual queues. Let $\sigma(j,t)$ denote the machine to which job $j$ is assigned at time $t$. Therefore, $i = \sigma(j,t)$ if and only if $j \in J_i(t)$. Let $\pi_j(t)$ denote the rank of $j$ in the virtual queue of $i = \sigma(j,t)$, where the head of the virtual queue has rank $1$ and the tail of the queue has rank $|J_i(t)|$.  Let $\S_j(t)$ denote the set of jobs ahead of job $j \in J_i(t)$ in the virtual queue of machine $i$. In other words 
$$\S_j(t) = \left\{ j' \;|\; \sigma(j',t) = \sigma(j,t) \mbox{ and } \pi_{j'}(t) < \pi_j(t) \right\}$$
Let $\W_j(t) = \sum_{j' \in \S_j(t)} w_{j'}$ denote the total weight of jobs ahead of job $j$ in its virtual queue.

\paragraph{Rate Assignment.} At time instant $t$, the total processing rate of the machine $i$ is divided among the jobs in $J_i(t)$ as follows. Job $j \in J_i(t)$ is assigned processing power $\nu_j(t)$ as follows:

\begin{equation}
	\label{eqn:jobshare}
	\nu_j(t)  :=  \eta \cdot \frac{(\W_{j}(t)+w_{j})^{k+1} - \W_j^{k+1} }{W(i,t)^{k+1}} 
\end{equation}
\noindent

The rate at which job $j \in J_i(t)$ is processed at time $t$ is therefore $\ell_{ij} \ \nu_j(t)$. Note that $\sum_{j \in J_i(t)} \nu_j(t) = \eta$ at all time instants $t$ and for all machines $i$. Note that if $k = 0$, this is exactly weighted round robin. As $k$ becomes larger, this gives higher rate to jobs in the tail of the queue, taking the weights $w_j$ into account. This ensures that small jobs arriving later do not contribute too much to the flow-time, hence reducing the speed augmentation.   One important property of \slaps(k) is that if a new job is added to the tail of the virtual queue, then all the old jobs are slowed down by the same factor. This is one of the important characteristics of weighted round robin which ensures that for a pair of jobs weighted delay induced by each other is exactly same. $\slaps(k)$ preserves this property to a factor of $O(k)$, and this will be crucial to our analysis. 

We note that using the natural setting of $k=0$ (weighted round robin) gives a competitive algorithm with speedup $(2+\epsilon)$, and this is tight even for a single machine. We use a larger value of $k$ to reduce the amount of speed augmentation needed. (We believe that $\slaps(k)$ gives a black-box reduction from any $(2+\epsilon)$-speed algorithm using \wrr into a scalable algorithm.)

\paragraph{Arrival Policy.} The behavior of the policy is the same when a job $j$ either arrives to the system and chooses machine $i$, or migrates from some other machine to machine $i$ at time $t$. In either case, the job $j$ is placed at the {\em tail} of the virtual queue. In other words, if $J_i(t^-)$ is the set of jobs just before the arrival of job $j$, then we set $\sigma(j,t) = i$ and $\pi_j(t) = |J_i(t^-)| + 1$. Therefore, the virtual queue sorts the jobs in order in which they arrive {\em onto this machine}. Since a job could also arrive due to a migration, this is {\em not the same} as ordering on the $r_j$ -- a job with smaller $r_j$ that migrates at a latter point onto machine $i$ will occupy a relatively later position in its virtual queue.

\paragraph{Departure Policy.} If job $j$ departs from machine $i$ either due to completion or due to migrating to a different machine, the job simply deletes itself from $i$'s virtual queue, keeping the ordering of other jobs  the same. In other words, for all jobs $j' \in J_i(t)$ with $\pi_{j'}(t) > \pi_j(t)$, the value $\pi_{j'}(t)$ decreases by 1.

\subsubsection{Virtual Utility of Jobs and Selfish Migration}
The virtual queues define a {\em virtual utility} of job as follows. Let $j \in J_i(t)$ at time $t$. Then its virtual utility is defined as:
$$ \phi(j,t) = \frac{\ell_{ij}}{\W_j(t) + w_j}$$

We interpret this utility as follows: The inverse of this quantity will be %(to within a constant factor) 
roughly in proportion to the marginal increase in instantaneous weighted delay that job $j$ induces on jobs $\S_j(t)$ that are ahead of it in its virtual queue, and their contribution to the weighted delay of $j$. We will establish this in the Delay Lemmas below. This marginal increase is exactly what we need in order to define dual variables in our proof, and in some sense, the virtual utility is defined keeping this in mind.

At every time instant, job $j \in J_i(t)$ behaves as follows: If it were to migrate to machine $d \neq i$, it would be placed at the tail of $d$'s queue and would obtain virtual utility $\frac{\ell_{dj}}{W(d,t) + w_j}$. If this quantity is larger than $\phi(j,t)$, then job $j$ migrates to machine $d$. This leads to the corresponding changes to the virtual queues of machine $i$ (job $j$ is deleted), machine $d$ (job $j$ is added to the tail), and the virtual utility $\phi(j,t)$ of job $j$ (which is set to $\frac{\ell_{dj}}{W(d,t) + w_j}$).  At every time instant $t$, this defines a game on the jobs, and starting with the configuration at the previous step, the jobs simulate sequential best response dynamics, where they sequentially migrate to improve virtual utility, till the system reaches a Nash equilibrium. In this configuration, each job is on a machine that is locally optimal for $\phi(j,t)$. 

Note that if a job departs from a machine, the virtual utilities of other jobs on that machine either stay the same or increase. Further, if a job migrates to a machine, it is placed on the tail of the virtual queue, so that the virtual utilities of other jobs on the machine remain the same. This shows that sequential best response dynamics guarantees that the virtual utilities of {\em all} jobs are monotonically non-decreasing with time, converging to a Nash equilibrium.  (Note that jobs don't actually need to execute best response dynamics since they can directly migrate to the machines corresponding to the resulting Nash equilibrium.)

When a new job arrives to the system, it simply chooses the machine $i$ which maximizes its virtual utility, $\frac{\ell_{ij}}{W(i,t) + w_j}$, where $W(i,t)$ is the weight of jobs assigned to $i$ just before the arrival of job $j$. This completes the description of the algorithm.

The following lemma is an easy consequence of the description of the algorithm.
\begin{lemma}
\label{lem:increase}
For all jobs $j$, $\phi(j,t)$ is non-decreasing over the interval $t \in [r_j, C_j]$.
\end{lemma}
 
\subsection{Analysis of {\sc SelfishMigrate}}
We first write a linear programming relaxation of the problem \lpp described below which was first given by~\cite{AnandGK12,GargK07}. It has a variable $x_{ijt}$ for each machine $i \in [m]$, each job $j \in [n]$ and each unit time slot $t \geq r_{j}$. If the machine $i$ processes the job $j$ during the whole time slot $t$, then this variable is set to $1$. The first constraint says that every job has to be completely processed. The second constraint says that a machine cannot process more than one unit of jobs during any time slot. Note that the LP allows a job to be processed simultaneously across different machines.

\[ \mbox{Min} \ \ \ \sum_{j} \sum_{i} \sum_{t \geq r_{j}}  \left(\frac{ \ell_{ij}(t- r_j)}{p_{j}} + 1 \right) \cdot w_{j}\cdot x_{ijt}  \qquad \tag{\lpp} \label{primal2} \]
\[ \begin{array}{rcllr}
\displaystyle 	\sum_{i} \sum_{t \geq r_{j}} \frac{\ell_{ij} \cdot x_{ijt}}{p_{j}} &\geq& 1 \qquad &\forall j   \\ 
\displaystyle	\sum_{j \, : \, t \geq r_{j}} x_{ijt} &\leq& 1 &\forall i,    t      \\
\displaystyle	x_{ijt} &\geq& 0  &\forall i,  j,  t \, : \, t \geq r_{j} \qquad 
\end{array}  \]

It is easy to show that the above LP lower bounds the optimal flow-time of a feasible schedule within factor 2.  We use the dual fitting framework to analyze {\sc SelfishMigrate}. We write the dual of $\lpp$ as,

\[ \mbox{Max} \ \ \ \sum_{j} \alpha_j -   \sum_{i} \sum_t \beta_{it} \qquad \tag{\lpd} \label{Dual} \]
\[ \begin{array}{rcllr}
\displaystyle \;\;\;\;	\frac{\ell_{ij} \cdot \alpha_j}{p_{j}} - \beta_{it} &\leq& \displaystyle \frac{w_{j}\ell_{ij}(t- r_j)}{p_{j}} + w_{j} \qquad &\forall  i,j, t \, : \, t \geq r_{j}  \label{dual-constraint}\\ 
\displaystyle  \;\;\;\;	\alpha_j &\geq& 0  &\forall   j \\ 
\displaystyle  \;\;\;\;	\beta_{it} &\geq& 0  &\forall   i,t  
\end{array}  \]

We will show that there is a feasible solution to $\lpd$ that has objective $O(\epsilon^2)$ times the total weighted flow-time of {\sc SelfishMigrate}, provided we augment the speed by $\eta = (1 + 3 \epsilon)$. From now on, we will assume that each processor in {\sc SelfishMigrate} has $\eta$ extra speed when processing jobs. 

\subsubsection{Instantaneous Delay and Setting Dual Variables}
Recall that each machine runs \slaps$(k)$ with $k = 1/\epsilon$, and we assume without loss of generality that $1/\epsilon$ is an integer.  We define the {\em instantaneous weighted delay} induced by job $j$ on jobs ahead of $j$ in its virtual queue (the set $\S_j(t)$) as follows:
$$ \delta_j(t) = \frac{1}{\eta} \left(\sum_{j' \in \S_j(t)}  \left( w_{j'} \cdot \nu_{j}(t) + w_{j}\cdot\nu_{j'}(t) \right) + w_{j} \cdot \nu_{j}(t) \right) $$
This quantity sums the instantaneous weighted delay that $j' \in \S_j(t)$ induces on $j$ and vice versa, plus the delay seen by $j$ due to itself. Note that $\delta_j(t)$ is equal to $\frac{1}{\eta} \left((\W_j(t) + w_j)\cdot \nu_{j}(t) + w_j \cdot \sum_{j'\in \S_j(t)}\nu_{j'}(t) \right)$.

Define,
$$ \Delta_j = \int_{t=r_j}^{C_j} \delta_j(t) dt$$ 
as the cumulative weighted delay induced by $j$ on jobs ahead of it in its virtual queue and vice versa. Note that the set $\S_j(t)$ changes with $t$ and can include jobs that are released after job $j$.  It is an easy exercise to check that $\sum_j w_j F_j = \sum_j \Delta_j$. Our way of accounting for weighted delay is therefore a significant departure from previous work that either keeps $\S_j(t)$ the same for all $t$ (clairvoyant algorithms), or preserves orderings based on arrival time.  

\medskip
We now perform the dual fitting. We set the variables of the \ref{Dual} as follows. We set $\beta_{it}$ proportional to the total weight of jobs alive on machine $i$ at time $t$, {\em i.e.}, $\beta_{it} = \frac{1}{k+3} W(i,t)$. We set $\alpha_j = \frac{1}{k+2} \Delta_j$, {\em i.e.}, proportional to the cumulative weighted delay induced by $j$ on jobs ahead of it in its virtual queue.

We first bound the dual objective as follows (noting $k = 1/\epsilon$ and $\eta = 1+3 \epsilon$):
\begin{eqnarray}
\sum_{j} \alpha_j - \sum_{i,t}\beta_{it} &=&  \sum_{j} \frac{\Delta_j}{k+2} - \sum_{i,t} \frac{W(i,t)}{k+3} \nonumber 
=  \epsilon \left(\sum_{j} \frac{\Delta_j}{1+2\epsilon} - \sum_{i,t} \frac{W(i,t)}{1 + 3\epsilon} \right) \nonumber\\ 
&=& \epsilon \cdot \sum_j w_jF_j \cdot \left(\frac{1}{1+2\epsilon} - \frac{1}{1+3\epsilon}\right) = O(\epsilon^2) \sum_j w_jF_j
\end{eqnarray}

In the rest of the analysis, we show that this setting of dual variables satisfies the dual constraints.

\subsubsection{Delay Lemmas}
The dual constraints need us to argue about the weighted delay induced by $j$ till any point $t$. For this purpose, we define for any $t^* \in [r_j,C_j]$ the following:
$$ \Delta^1_j(t^*)  = \int_{t=r_j}^{t^*} \delta_j(t) dt  \qquad \mbox{and} \qquad \Delta^2_j(t^*)  = \int_{t=t^*}^{C_j} \delta_j(t) dt$$

The following propositions have elementary proofs which have been omitted.

\begin{proposition}
	\label{prop:lower}
	Consider any integer $k \geq 0$, and $ 0 \leq w \leq 1$, then $(1-w)^k \geq 1 -kw$.
\end{proposition}

\begin{proposition}
	\label{prop:upper}
%	Consider any integer $k \geq 0$, and $ 0 \leq w \leq 1$, then $(1 + w)^k \geq 1 + kw$.
% xxx
	Consider any integer $k \geq 0$, and $w, w' \geq 0 $, then $(w + w')^k \geq w^k + k w^{k-1} w'$.
\end{proposition}

The first Delay Lemma bounds the quantity $\Delta^1_j(t^*)$ as follows:
\begin{lemma}[\textbf{First Delay Lemma}]
\label{firstdelay}
For any time instant $t^* \in [r_j,C_j]$ and for any job $j$, 
$$ \Delta^1_j(t^*) \leq \quad(k+2) \cdot w_{j}\cdot(t^*-r_j) $$
\end{lemma}
\begin{proof} 
{\small
\begin{eqnarray*}
\Delta^1_j(t^*) &=& \frac{1}{\eta} \int^{t^*}_{t = r_j} \left( \nu_{j}(t) \cdot (\W_j(t) + w_j) + w_{j} \cdot \left(\displaystyle \sum_{j' \in \S_{j}(t)} \nu_{j'}(t) \right) \right) dt  \\
% xxx
&\leq& \frac{1}{\eta} \int^{t^*}_{t = r_j} \left( \eta \cdot  \frac{ (\W_{j}(t) + w_j)^{k+1} - \W_{j}(t)^{k+1}} {W(\sigma(j,t),t)^{k+1}} \cdot (\W_{j}(t) + w_j)+ w_{j} \cdot \eta \right) dt  \qquad \mbox{[Definition of $\nu$]} \\
&\leq& \int^{t^*}_{t = r_j} \left( \frac{(\W_{j}(t) + w_j)^{k+1} - \W_{j}(t)^{k+1}} {(\W_{j}(t) + w_j)^{k+1}} \cdot (\W_{j}(t) + w_j) + w_{j} \right) dt \qquad \mbox{[Since $W(\sigma(j,t),t) \geq \W_j(t) {\huge +} w_j$]} \\
&=& \int^{t^*}_{t = r_j} \left( \left (1 - \left(1 - \frac{w_j}{\W_{j}(t) + w_j}\right)^{k+1} \right) \cdot (\W_{j}(t) + w_j) + w_j \right )dt\\
&\leq& \int^{t^*}_{t = r_j} \left( \frac{w_{j} \cdot (k+1)}{(\W_{j}(t) + w_j)} \cdot (\W_{j}(t) + w_j) + w_{j} \right) dt  \qquad \mbox{[Proposition~\ref{prop:lower}]} \nonumber \\
&=& (k+2)\cdot w_{j}\cdot (t^*-r_j) 
\end{eqnarray*}}
\end{proof}

Let $p_j(t^*) = \int^{C_j}_{t = t^*}\ell_{\sigma(j,t)j} \cdot \nu_{j}(t) dt$ denote the residual size of job $j$ at time $t^*$. The second Delay Lemma states that total marginal increase in the algorithm's cost due to job $j$ till its completion is upper bounded by the marginal increase in the algorithm's  cost  if the job $j$ stays on machine $\sigma(j,t^*)$ till its completion. However, as noted before, marginal increase in the cost of the algorithm on a single machine is inversely proportional to the job'virtual speed. The proof of the second Delay Lemma hinges crucially on the fact that a job selfishly migrates to a new machine only if its virtual utility increases.  In fact, the statement of this lemma implies the correctness of our setting of virtual utility.

\begin{lemma}[\textbf{Second Delay Lemma}]
\label{lem1}
For any time instant $t^* \in [r_j,C_j]$ and for any job $j$,  let $i^* = \sigma(j,t^*)$ denote the machine to which job $j$ is assigned at time $t^*$. Then:
$$ \Delta^2_j(t^*) \le  \frac{1}{\eta} \cdot\frac{k+2}{k+1} \cdot \frac{p_j(t^*)}{\phi(j,t^*)} \le \frac{1}{\eta} \cdot\frac{k+2}{k+1} \cdot p_j \cdot \frac{\W_j(t^*)+ w_j}{l_{i^*j}} $$
\end{lemma}
\begin{proof}
{\small
\begin{eqnarray*} \vspace{-5mm}
\Delta^2_j(t^*) &=& \frac{1}{\eta} \int^{C_j}_{t = t^*} \left( \nu_{j}(t) \cdot (\W_{j}(t) + w_j) + w_{j} \cdot \left(\displaystyle \sum_{j' \in \S_{j}(t)} \nu_{j'}(t) \right) \right) dt  \\
&=& \frac{1}{\eta}  \int^{C_j}_{t = t^*} \left( \nu_{j}(t) \cdot  (\W_{j}(t) + w_j) + \eta \cdot w_{j} \cdot \frac{\W_{j}(t)^{k+1}}{W(\sigma(j,t),t)^{k+1}} \right) dt   \\
&=& \frac{1}{\eta}  \int^{C_j}_{t = t^*} \nu_{j}(t) \cdot \left( \W_{ j}(t) + w_j + w_j \cdot \frac{\W_{j}(t)^{k+1}} {(\W_{j}(t) + w_j)^{k+1} - \W_{j}(t)^{k+1}} \right) dt  \nonumber \\
&\leq& \frac{1}{\eta}  \int^{C_j}_{t = t^*} \nu_{j}(t) \cdot \left( \W_{j}(t) + w_j + \frac{\W_{j}(t) + w_j}{k+1} \right) dt  \qquad \mbox{[Proposition~\ref{prop:upper}]}\nonumber \\
&=&  \frac{1}{\eta} \cdot \frac{k+2}{k+1} \cdot \int^{C_j}_{t = t^*} \ell_{\sigma(j,t)j}\cdot \nu_{j}(t)  \cdot \frac{1}{\phi(j,t)} dt \nonumber \\
&\leq&  \frac{1}{\eta} \cdot \frac{k+2}{k+1} \cdot \frac{1}{\phi(j,t^*)} \cdot \int^{C_j}_{t = t^*} \ell_{\sigma(j,t)j}\cdot \nu_{j}(t) dt \qquad \mbox{[Lemma~\ref{lem:increase}]}\nonumber \\
&=&  \frac{1}{\eta} \cdot\frac{k+2}{k+1} \cdot \frac{p_j(t^*)}{\phi(j,t^*)}   \le \frac{1}{\eta} \cdot\frac{k+2}{k+1} \cdot p_j \cdot \frac{\W_j(t^*) + w_j}{l_{i^*j}} 
\end{eqnarray*}}
\end{proof}

Note that the previous two lemmas imply the following by summation: 
\begin{lemma}
\label{lem:delay}
For any time instant $t \in [r_j,C_j]$ and job $j$ that is assigned to machine $i^* = \sigma(j,t)$, we have:
$$ \Delta_j = \Delta^1_j(t) + \Delta^2_j(t) \le (k+2) \cdot w_{j}\cdot(t-r_j) + \frac{1}{\eta} \cdot\frac{k+2}{k+1} \cdot p_j \cdot \frac{\W_j(t) + w_j}{l_{i^*j}} $$
\end{lemma}

\subsubsection{Checking the Feasibility of Constraints}
Now it  remains to prove that  constraints of \ref{Dual} are satisfied.  To see this, fix job $j$ and time instant $t$. We consider two cases.

\noindent \textbf{Case 1:} Machine $i = \sigma(j,t)$. Then
\begin{eqnarray}
\alpha_j - \frac{p_j}{\ell_{ij}}\beta_{it} &=& \frac{\Delta_j}{k+2} - \frac{p_j}{\ell_{ij}}\cdot \frac{W(i,t)}{k+3}  \nonumber \\
&\leq& w_{j} \cdot (t-r_j) + \frac{p_j}{\eta \cdot (k+1)} \cdot \frac{\W_{j}(t)+ w_{j}}{\ell_{ij}} - \frac{p_j}{\ell_{ij}}\cdot \frac{W(i,t)}{k+3}  \qquad \mbox{[Lemma~\ref{lem:delay}]} \nonumber \\
&\leq& w_{j} \cdot (t - r_j) \qquad [\mbox{since } \eta = 1+3\epsilon,k = 1/\epsilon]\nonumber
\end{eqnarray}

\noindent \textbf{Case 2:} Machine $i \neq \sigma(j,t)$. Then
\begin{eqnarray}
\alpha_j - \frac{p_j}{\ell_{ij}}\beta_{it} &=& \frac{\Delta_j}{k+2} - \frac{p_j}{\ell_{ij}}\cdot \frac{W(i,t)}{k+3} \nonumber \\
&\leq& w_{j} \cdot (t-r_j) + \frac{1}{\eta} \cdot \frac{p_j}{k+1}\cdot \frac{\W_{j}(t) + w_{j}}{\ell_{\sigma(j,t)j}} - \frac{p_j}{\ell_{ij}}\cdot \frac{W(i,t)}{k+3} \qquad \mbox{[Lemma~\ref{lem:delay}]} \nonumber \\
& \leq & w_{j} \cdot (t-r_j) + \left (\frac{p_j}{\eta\cdot(k+1)}\frac{W(i,t)+w_{j}}{\ell_{ij}} - \frac{p_j}{k+3} \cdot \frac{W(i,t)}{\ell_{ij}} \right)  \nonumber \\
&\leq& w_{j}\cdot(t-r_j) + \frac{w_{j}p_j}{ \ell_{ij}} \qquad [\mbox{since } \eta = 1+3\epsilon, k = 1/\epsilon]\nonumber \nonumber
\end{eqnarray}

The penultimate inequality follows since the machine $\sigma(j,t)$ maximizes the virtual utility of job $j$ at time $t$. Therefore, the dual constraints are satisfied for all time instants $t$ and all jobs $j$, and we derive that {\sc SelfishMigrate} is $(1+\eps)$-speed augmentation, $O(1/\epsilon^2)$-competitive against $\lpp$, completing the proof of the first part of Theorem \ref{thm:flow-no-scalable}.

\smallskip
\noindent \textbf{Polynomial Time Algorithm and Minimizing Reassignments.}
A careful observation of the analysis reveals that to satisfy dual constraints, each job need not be on the machine which gives the highest virtual utility. We can change the policy {\sc SelfishMigrate} so that a job migrates to a different machine only if its virtual utility increases by a factor of at least $(1+\epsilon)$. Note that this does not change the monotonicity properties of the virtual utility of a job, hence the entire analysis follows (with the speed augmentation $\eta$ increased by a factor of $1+ \eps$). This also implies that for any job $j$, the total number of migrations is at most $(\log_{(1+\epsilon)} W + \log_{(1+\epsilon)}n)$, where $W$ is the ratio of the maximum weight of all jobs to the minimum weight. Omitting the simple details, we complete the proof of Theorem \ref{thm:flow-no-scalable}.

\section{Weighted Flow-time and Energy for Arbitary Power Functions}
\label{sec:energy}
In this section we present a simple extension to {\sc SelfishMigrate} to get a scalable algorithm for minimizing the sum of weighted flow-time and energy for arbitrary power functions. The problem formulation is the same as in Section~\ref{sec:unrelated}, with an added feature.  Each machine $i$ can be run at a variable rate $S(i,t)$ by paying an energy cost of $f_i(S(i,t))$, where $f_i$ is a machine dependent, convex increasing function (also called as power function). The rate $S(i,t)$ can be partitioned among the jobs $J_i(t)$, so that $\sum_{j \in J_i(t)} \nu_j(t) \le S(i,t)$. As before, job $j \in J_i(t)$ runs at speed $q_j(t) = \nu_j(t) \times \ell_{ij}$. 

As in Section~\ref{sec:unrelated}, we define the completion time $C_j$ of job $j$ to satisfy $\int_{t=r_j}^{C_j} q_j(t) dt = p_j$. As before, preemption and migration of jobs are allowed without any penalty, but  each job must be assigned to a single machine at every instant of time. The scheduler is not aware of the processing lengths $p_j$.  Our objective is to minimize sum of weighted flow-time and energy consumed: 
$$ \mbox{Objective } =   \sum_{j} w_j F_j + \sum_{i}\int_{t}f_i(S(i,t)) dt$$

In a resource augmentation analysis, we assume that the online algorithm  gets $(1+\epsilon)$ more speed, for any $\epsilon > 0$, for {\em consuming the same energy.} Alternatively, the offline benchmark has to pay a cost of $f_i((1+\epsilon)s)$ if it runs machine $i$ at a rate of $s$. Speed augmentation is required to achieve meaningful competitive ratios for the case of arbitrary power functions. To elaborate on this point, consider a function $f(s)$ that takes an infinitesimal value  in the interval $0 \leq s \leq 1$ and sharply increases when $s > 1$. For such a power function, any competitive online scheduler has to be optimal at each instant of time unless we give it more resources.  A scalable algorithm in the speed augmentation setting implies algorithms with small competitive ratios when the energy cost function can be approximated by polynomials. In particular, the result translates to an $O(\gamma^2)$-competitive algorithm (without any resource augmentation) when the power function is $f_i(s) = s^{\gamma}$. 

Let $g_i$ be the inverse of power function $f_i$. Note that $g$ is an increasing concave function. Before we describe our algorithm, we make the following simple observation regarding concave functions.

\begin{proposition}
\label{cpl1}
For any increasing concave function $g$ with $g(0) = 0$, $\frac{g(w)}{w}$ is decreasing in $w$.
\end{proposition}

\subsection{The {\sc SelfishMigrate-Energy} Algorithm}
Our algorithm {\sc SelfishMigrate-Energy}, is very similar to the algorithm {\sc SelfishMigrate}, and we only outline the differences with Section~\ref{sec:unrelated}. The most important difference is the policy that sets the speeds of the machines.  

\paragraph{Speed Scaling Policy:} We set the speed of machine $i$ at time $t$, denoted by $S(i,t)$, such that the total energy cost is equal to the total weight of jobs at time $t$.
\begin{eqnarray}
 f_i(S(i,t)) = W(i,t)  \qquad \mbox{or equivalenetly,} \qquad  S(i,t) = g_i(W(i,t))
\end{eqnarray}

This is same as the speed scaling policy used in \cite{NH2009,AnandGK12}. Our speed scaling policy easily implies that the total energy cost of the schedule is equal to the weighted flow-time. Hence, we will only be concerned with the total weighted flow-time of our schedule.

\paragraph{Single Machine Policy.} The remaining components of the algorithm remain similar to {\sc SelfishMigrate} algorithm. We briefly mention the differences.

Each machine runs \slaps$(k)$ where $k = \frac{1}{\epsilon}$. In this policy, the notions of virtual queues, rank of a job, and the arrival and departure policies (with associated notation) -- remain the same. In particular, a job that arrives or migrates to a machine are placed at the tail of the virtual queue and assigned the highest rank on the machine. At time instant $t$, the total processing rate $S(i,t) = g_i(W(i,t))$ of the machine $i$ is divided among the jobs in $J_i(t)$ as follows. 
\begin{equation}
  \label{eqn:jobshare2}
  \nu_j(t)  :=  g_i(W(i,t)) \cdot \frac{(\W_{j}(t)+w_{j})^{k+1} - \W_j^{k+1} }{W(i,t)^{k+1}} 
\end{equation}
\noindent
As before, this implies job $j \in J_i(t)$ is processed at rate $\ell_{ij} \ \nu_j(t)$. 

\paragraph{Virtual Utility of Jobs and Selfish Migration.} Consider a job $j \in J_i(t)$ at time $t$. Its virtual utility is defined as:
$$ \phi(j,t) = g_i(\W_j(t) + w_j) \cdot \frac{\ell_{ij}}{\W_j(t) + w_j}$$
Using this virtual utility, the jobs perform sequential best response dynamics, migrating to other machines if it improves its virtual utility. As before, this leads to a Nash equilibrium every step.\footnote{Note that unlike Section~\ref{sec:unrelated} where migrations can only happen on departure of jobs from the system, migrations  can now also happen on arrival of jobs into the system.} If a job moves out of a machine, the weights $\W_{j'}(t)$ of other jobs on the machine either stay the same or decrease. Using Proposition~\ref{cpl1}, this implies the virtual utility of other jobs on the machine either remains the same or increases. Therefore, similar to Lemma \ref{lem:increase}, we easily get the monotonicity of the virtual utilities of jobs.
\begin{lemma}
\label{lem:powerincrease}
For all jobs $j$, $\phi(j,t)$ is non-decreasing over the interval $t \in [r_j, C_j]$.
\end{lemma}

\subsection{Analysis of {\sc SelfishMigrate-Energy}}
Since our speed scaling policy ensures that total weighted flow-time of the schedule is equal to the energy cost, we focus on bounding the total weighted flow-time of jobs. 

\paragraph{Convex Programming Relaxation.} Consider the following convex programming relaxation for the problem due to \cite{AnandGK12, Devanur014}. In this relaxation, there is a variable $s_{ijt}$ which indicates the speed at which job $j$ is processed at time $t$ on machine $i$. The constraints of \cpp state that each job needs to be completely processed and the speed of machine $i$ at time $t$ is equal to the sum of individual speeds of the jobs.

We now give a brief explanation on why the objective function lower bounds the optimal schedule within a factor of 2. See \cite{Devanur014, AnandGK12} for a complete proof of this claim. The first term in the objective function lower bounds the weighted flow-time of jobs and is similar to the term in the \lpp. The second term corresponds to the energy cost of the schedule. Here we use the fact that we analyse our algorithm in the resource augmentation model. Hence, the offline benchmark pays a cost of $f_i((1+\epsilon)\cdot s_{it})$ for running at a speed of $s_{it}$. The third term is a lowerbound on the total cost any optimal solution has to pay to schedule a job $j$, assuming that the job $j$ is the only job present in the system. This term is needed, as we do not explicitly put any constriants to forbid simultaneous processing of job $j$ across machines. Clearly, without this term, \cpp has a huge integrality gap as a single job can be processed to an extent of $\frac{1}{m}$ simulatenously on all machines. The function $f^*_i$ is the Legendre-Fenchel conjugate of function $f$ and is defined as $f^*(\beta) = max_{s} \{ s \cdot \beta - f(s) \}$. See \cite{Devanur014} for more details.

\[ \mbox{Min} \ \ \ \sum_{j} \sum_{i} \int_{t \geq r_{j}}  (t-r_j) \cdot \frac{\ell_{ij}w_{j}}{p_{j}}\cdot s_{ijt} \hspace{1mm} dt + \sum_{i} \int_{t} f_i((1+3\epsilon)s_{it}) \hspace{1mm} dt+  \sum_{i} \sum_{j} \int_{t \geq r_j} (f_i^*)^{-1}(w_j)s_{ijt} \hspace{1mm} dt \qquad \tag{\cpp} \label{primal3} \]
\[ \begin{array}{rcllr}
\displaystyle   \sum_{i} \int_{t \geq r_{j}} \frac{\ell_{ij} \cdot s_{ijt}}{p_{j}} &\geq& 1 \qquad &\forall j   \\ 
\displaystyle \sum_{j \, : \, t \geq r_{j}} s_{ijt} &=& s_{it} &\forall i,    t      \\
\displaystyle s_{ijt} &\geq& 0  &\forall i,  j,  t \, : \, t \geq r_{j} \qquad 
\end{array}  \]

We write the dual of \cpp following the framework given in \cite{Devanur014}. Similar to the dual of weighted flow-time, we have a variable $\beta_{it}$ for each machine $i$ and time instant $t$, and a variable $\alpha_j$ for each job $j$. 

\[ \mbox{Max} \ \ \ \sum_{j} \alpha_j -   \sum_{i} \int_t f_i^*\left(\frac{\beta_{it}}{1+3\epsilon} \right) \hspace{1mm} dt \qquad \tag{\cpd} \label{Dual2} \]
\[ \begin{array}{rcllr}
\displaystyle \;\;\;\;  \frac{\ell_{ij} \alpha_j}{p_{j}} - \beta_{it} &\leq& \displaystyle \frac{\ell_{ij} w_{j}}{p_{j}}(t- r_j) + (f_i^*)^{-1}(w_j) \qquad &\forall  i,j, t \, : \, t \geq r_{j}  \label{dual-constraint2}\\ 
\displaystyle  \;\;\;\; \alpha_j &\geq& 0  &\forall   j \\ 
\displaystyle  \;\;\;\; \beta_{it} &\geq& 0  &\forall   i,t  
\end{array}  \]

\noindent We need the following simple observation regarding $f$ and $f^*$ for the rest of the analysis.
\begin{lemma}
\label{convex:property}
For any increasing strictly convex function $f$ with $f(0) = 0$, let $g = f^{-1}$ and let $f^*$ be its Legendre-Fenchel conjugate. Then $f^* \left( \frac{w}{g(w)} \right) \leq w$, or $\frac{w}{g(w)} \le (f^*)^{-1}(w)$.
\end{lemma}
\begin{proof}
 From the definition of $f^*$, it is enough to show that for all $x,w$, we have: $\frac{wx}{g(w)} - f(x) \leq w$. Consider the following two cases. If $x \leq g(w)$, then the condition is trivially true. Now consider the case when $x \ge g(w)$. It follows that $f(x) / x$ is non-decreasing from convexity of $f$ and the fact $f(0) = 0$. Hence we have $f(x) / x \geq f(g(w)) / g(w) = w / g(w)$, which completes the proof. 
\end{proof}

\paragraph{Instantaneous Delay and Setting  Variables.} Recall that each machine runs \slaps$(k)$ with $k = 1/\epsilon$. Let $i^* = \sigma(j,t)$ denote the machine to which job $j$ is assigned at time $t$.  We define the {\em instantaneous weighted delay} induced by job $j$ on jobs ahead of $j$ in its virtual queue (the set $\J_j(t)$) and the weighted delay incurred by job $j$ due to jobs in $\J_j(t)$ as follows:
$$ \delta_j(t) =  \frac{1}{g_{i^*}(W(i^*,t))} \cdot \left(\sum_{j' \in \S_j(t)}  \left( w_{j'} \cdot \nu_{j}(t) + w_{j}\cdot\nu_{j'}(t) \right) + w_{j} \cdot \nu_{j}(t) \right) $$
Define  $ \Delta_j = \int_{t=r_j}^{C_j} \delta_j(t) dt$ as the cumulative weighted delay induced by $j$ on jobs ahead of it in its virtual queue. Again, it is easy to see that $\sum_j w_j F_j = \sum_j \Delta_j$. 

\medskip
We now perform the dual fitting. We set $\beta_{it} = \frac{1}{k} \cdot \frac{W(i,t)}{g_i(W(i,t))}$ and  $\alpha_j = \frac{1}{k+2} \Delta_j$. As before, we have:

\begin{eqnarray*}
\sum_{j} \alpha_j - \sum_{i,t}f_i^*\left(\frac{\beta_{it}}{1+3\epsilon}\right) &\geq&  \sum_{j} \frac{\Delta_j}{k+2} - \sum_{i,t} \frac{\epsilon \cdot W(i,t)}{1+3\epsilon} \qquad[\text{by convexity of $f^*_i$ and Lemma \ref{convex:property}}]\nonumber \\
&=&  \epsilon \left(\sum_{j} \frac{\Delta_j}{1+2\epsilon} - \sum_{i,t} \frac{W(i,t)}{1 + 3\epsilon} \right) \nonumber\\ 
&=& \epsilon \cdot \sum_j w_jF_j \cdot \left(\frac{1}{1+2\epsilon} - \frac{1}{1+3\epsilon}\right) = O(\epsilon^2) \sum_j w_jF_j
\end{eqnarray*}
Since the energy cost of {\sc SefishMigrate-Energy} is equal to total weighted flow-time, we get a $O(1/\epsilon^2)$-competitive algorithm with a speed-augmentation of $(1+\epsilon)$. 

\paragraph{Delay Lemmas.} Similar to the delay lemmas for total weighted flow-time, we establish corresponding delay lemmas. Define for any $t^* \in [r_j,C_j]$ the following:
$$ \Delta^1_j(t^*)  = \int_{t=r_j}^{t^*} \delta_j(t) dt  \qquad \mbox{and} \qquad \Delta^2_j(t^*)  = \int_{t=t^*}^{C_j} \delta_j(t) dt$$
The following lemma bounds the quantity $\Delta^1_j(t^*)$; the proof is  identical to the proof of Lemma \ref{firstdelay}:
\begin{lemma}
\label{powerDL2}
For any time instant $t^* \in [r_j,C_j]$ and for any job $j$, 
$$ \Delta^1_j(t^*) \leq \quad(k+2) \cdot w_{j}\cdot(t^*-r_j) $$
\end{lemma}

Let $p_j(t^*) = \int^{C_j}_{t = t^*}\ell_{\sigma(j,t)j} \cdot \nu_{j}(t) dt$ denote the residual size of job $j$ at time $t^*$. Next we establish the corresponding Second Delay Lemma.

\begin{lemma}
\label{powerDL1}
For any time instant $t^* \in [r_j,C_j]$ and for any job $j$,  let $i^* = \sigma(j,t^*)$ denote the machine to which job $j$ is assigned at time $t^*$. Then:
$$ \Delta^2_j(t^*) \le  \frac{k+2}{k+1} \cdot \frac{p_j(t^*)}{\phi(j,t^*)} \le  \frac{k+2}{k+1} \cdot  \frac{p_j}{l_{i^*j}} \cdot \frac{\W_j(t^*)+ w_j}{g_{i^*}(\W_j(t^*)+ w_j)} $$
\end{lemma}
\begin{proof}
{\small
\begin{eqnarray*}
\Delta^2_j(t^*) &=&  \int^{C_j}_{t = t^*} \frac{1}{g_{i^*}(W(i^*,t))} \cdot \left( \nu_{j}(t) \cdot (w_j + \W_{j}(t)) + w_{j} \cdot \left(\displaystyle \sum_{j' \in \J_{j}(t)} \nu_{j'}(t) \right) \right) dt   \nonumber \\
&=& \int^{C_j}_{t = t^*} \frac{1}{g_{i^*}(W(i^*,t))} \cdot \nu_{j}(t) \cdot \left( (w_j + \W_{j}(t)) + \frac{w_{j} \cdot \W_{j}(t)^{k+1}} {(\W_{j}(t) + w_j)^{k+1} - \W_{j}(t)^{k+1}} \right) dt \qquad \mbox{[From the def of $\nu$]}    \nonumber \\
&\leq& \int^{C_j}_{t = t^*} \frac{1}{g_{i^*}(W(i^*,t))} \cdot \nu_{j}(t) \cdot \left( (w_j + \W_{j}(t)) + \frac{w_j + \W_{j}(t)}{k+1} \right) dt  \qquad \mbox{[Proposition~\ref{prop:upper}]}\nonumber \\
&\leq& \frac{k+2}{k+1} \cdot \int^{C_j}_{t = t^*}  \frac{1}{g_{i^*}(\W_j(t) + w_j)} \cdot \nu_{j}(t) \cdot  (w_j + \W_{j}(t)) dt  \qquad \mbox{[$g$ is increasing]}\nonumber \\
&=&  \frac{k+2}{k+1} \cdot \int^{C_j}_{t = t^*} \ell_{i^*j}\cdot \nu_{j}(t)  \cdot \frac{1}{\phi(j,t)} dt \ \ \leq  \   \frac{k+2}{k+1} \cdot \frac{1}{\phi(j,t^*)} \cdot \int^{C_j}_{t = t^*} \ell_{i^*j}\cdot \nu_{j}(t) dt \qquad \mbox{[Lemma~\ref{lem:powerincrease}]}\nonumber \\
&=& \frac{k+2}{k+1} \cdot \frac{p_j(t^*)}{\phi(j,t^*)}  \le  \frac{k+2}{k+1} \cdot  \frac{p_j}{l_{i^*j}} \cdot \frac{\W_j(t^*)+ w_j}{g_{i^*}(\W_j(t^*)+ w_j)}
\end{eqnarray*}}
\end{proof}

From the previous two lemmas we get:
\begin{lemma}
\label{lem:powerdelay}
For any time instant $t \in [r_j,C_j]$ and job $j$ that is assigned to machine $i^* = \sigma(j,t)$, we have:
$$ \Delta_j = \Delta^1_j(t) + \Delta^2_j(t) \le (k+2) \cdot w_{j}\cdot(t-r_j) + \frac{k+2}{k+1} \cdot \frac{p_j}{l_{i^*j}} \cdot \frac{\W_j(t) + w_j}{g_{i^*}(\W_j(t) + w_j)} $$
\end{lemma}

\paragraph{Checking the Feasibility of \ref{Dual2} Constraints.}
Now it  remains to prove that  constraints of \ref{Dual2} are satisfied.  To see this, fix job $j$ and time instant $t$. We consider two cases.

\noindent \textbf{Case 1:} Machine $i = \sigma(j,t)$. Then
{\small
\begin{eqnarray*}
\alpha_j - \frac{p_j}{\ell_{ij}}\beta_{it} &=& \frac{\Delta_j}{k+2} - \frac{p_j}{\ell_{ij}}\cdot \frac{1}{k} \cdot \frac{W(i,t)}{g_i(W(i,t))}  \nonumber \\
&\leq& w_{j} \cdot (t-r_j) + \frac{p_j}{\ell_{ij}} \cdot \frac{1}{k+1} \cdot \frac{\W_{j}(t)+ w_{j}}{g_i(\W_{j}(t)+ w_{j})} - \frac{p_j}{\ell_{ij}}\cdot \frac{1}{k} \cdot \frac{W(i,t)}{g_i(W(i,t))}  \qquad \mbox{[Lemma~\ref{lem:powerdelay}]} \nonumber \\
&\leq& w_{j} \cdot (t - r_j) \qquad [\mbox{Proposition \ref{cpl1}}]\nonumber
\end{eqnarray*}}

\noindent \textbf{Case 2:} Machine $i \neq \sigma(j,t)$. Then
{\small
\begin{eqnarray*}
\alpha_j - \frac{p_j}{\ell_{ij}}\beta_{it} &=& \frac{\Delta_j}{k+2} - \frac{p_j}{\ell_{ij}}\cdot \frac{1}{k} \cdot \frac{W(i,t)}{g_i(W(i,t))} \nonumber \\
&\leq&w_{j} \cdot (t-r_j) + \frac{p_j}{\ell_{\sigma(j,t)j}} \cdot \frac{1}{(k+1)} \cdot \frac{\W_{j}(t)+ w_{j}}{g_{\sigma(j,t)}(\W_{j}(t)+ w_{j})} - \frac{p_j}{\ell_{ij}}\cdot \frac{1}{k} \cdot \frac{W(i,t)}{g_i(W(i,t))}  \qquad \mbox{[Lemma~\ref{lem:powerdelay}]} \nonumber \\
&\leq& w_{j} \cdot (t-r_j) + \frac{p_j}{\ell_{\sigma(j,t)j}} \cdot \frac{1}{(k+1)} \cdot \frac{\W_{j}(t)+ w_{j}}{g_{\sigma(j,t)}(\W_{j}(t)+ w_{j})} - \frac{p_j}{\ell_{ij}}\cdot \frac{1}{k} \cdot \frac{(W(i,t)+w_j-w_j)}{g_i(W(i,t) + w_j)} \hspace{2mm} [\mbox{since }  g \mbox{ is increasing}] \nonumber \\
& \leq & w_{j} \cdot (t-r_j) +  \frac{p_j}{\ell_{ij}} \cdot \frac{1}{(k+1)} \cdot \frac{w_{j}}{g_i(W(i,t)+ w_{j})} \qquad [\mbox{since $\sigma(j,t)$ maximizes virtual utility}]  \nonumber \\
& \leq & w_{j} \cdot (t-r_j) +  \frac{p_j}{\ell_{ij}} \cdot \frac{1}{(k+1)} \cdot \frac{w_{j}}{g_i(w_{j}) }  \qquad [\mbox{since }  g \mbox{ is increasing}] \nonumber  \\
&\leq& w_{j}\cdot(t-r_j) +\frac{p_j}{\ell_{ij}} \cdot (f_i^*)^{-1}(w_j) \qquad [\mbox{Lemma~\ref{convex:property}}]\nonumber \nonumber
\end{eqnarray*}}

\paragraph{Polynomial Power Functions.} As a corollary of the above result, we get a $O(\gamma^2)$-competitive algorithm when each machine follows the power function $s^{\gamma}$.

\begin{corollary}
There is $O(\gamma^2)$-competitive non-clairvoyant algorithm for minimizing weighted flow-time plus energy, when each machine follows a polynomial power function $f(s) = s^\gamma$.
\end{corollary}
\iffalse
\begin{proof}
Consider the convex programming formulation \cpp and set $\epsilon = 1/\gamma$. Then, the CP is only $O(1)$ more than any optimal solution. Thus, we get $O(\gamma^2)$-competitive algorithm.
\end{proof}
\fi

\section*{Acknowledgments}
We thank Naveen Garg, Benjamin Moseley, Ravishankar Krishnaswamy, Anupam Gupta,
Minghong Lin, Michael Nugent, and Neal Barcelo for initial discussions
about this line of research over the last few years.

\bibliographystyle{plain}
\bibliography{mmf}
\appendix

\end{document}